\documentclass{amsart}[10pt]
\usepackage{amsmath}
\usepackage{amssymb}
\usepackage{mathrsfs}
\usepackage{amsfonts}
\usepackage{graphicx}
\usepackage{tikz}
\usepackage{texdraw}
\usepackage{graphpap}
\usepackage{enumitem}
\usepackage{chngcntr}
\usepackage{wasysym}
\usepackage{color}
\usepackage{tikz}
\usepackage{todonotes}
\usepackage{comment}
\usepackage{soul}

\usetikzlibrary{decorations.markings}
\usetikzlibrary{spy}

\theoremstyle{plain}
\newtheorem{thm}{Theorem}[section]

\newtheorem{cor}[thm]{Corollary}
\newtheorem{lem}[thm]{Lemma}

\theoremstyle{remark}

\newtheorem*{rem}{Remark}

\numberwithin{equation}{section}

\newcommand{\ttn}{{\tt{n}}}
\newcommand{\tts}{{\tt{s}}}

\newcommand{\calA}{{\mathscr A}}
\newcommand{\calB}{{\mathcal B}}
\newcommand{\calQ}{{\mathcal Q}}
\newcommand{\calH}{{\mathscr H}}

\newcommand{\calW}{{\mathscr W}}

\newcommand{\Int}{\operatorname{Int}}

\newcommand{\trace}{\operatorname{trace}}
\newcommand{\fluct}{\operatorname{fluct}}
\newcommand{\erfc}{\operatorname{erfc}}

\newcommand{\R}{{\mathbb R}}

\newcommand{\C}{{\mathbb C}}
\newcommand{\E}{{\mathbb E}}
\newcommand{\bigO}{{\mathcal{O}}}

\newcommand{\eps}{{\varepsilon}}

\newcommand{\re}{\operatorname{Re}}
\newcommand{\im}{\operatorname{Im}}

\newcommand{\Prob}{\mathbb{P}}

\renewcommand{\d}{{\partial}}
\newcommand{\dbar}{\bar{\partial}}
\newcommand{\1}{\mathbf{1}}

\newcommand{\supp}{\operatorname{supp}}

\makeatletter
\newcommand*\bigcdot{\mathpalette\bigcdot@{.5}}
\newcommand*\bigcdot@[2]{\mathbin{\vcenter{\hbox{\scalebox{#2}{$\m@th#1\bullet$}}}}}
\makeatother

\allowdisplaybreaks

\begin{document}

\title[edge density $\sqrt{n}$-correction]{A formula for the edge density $\sqrt{n}$-correction for two-dimensional Coulomb systems}

\author{Yacin Ameur}
\address{Yacin Ameur\\
Department of Mathematics\\
Lund University\\
22100 Lund, Sweden}
\email{ Yacin.Ameur@math.lu.se}



\begin{abstract} In connection with recent work on smallest gaps, C.~Charlier proves that the 1-point function of a suitable planar Coulomb system $\{z_j\}_1^n$, in the determinantal case with respect to an external potential $Q(z)$, admits the expansion, as $n\to\infty$,
$$R_n\bigg(z_0+\frac t {\sqrt{2n\d\dbar Q(z_0)}}\nu(z_0)\bigg)=n\d\dbar Q(z_0)\frac {\erfc t}2+\sqrt{n\d\dbar Q(z_0)}\,C(z_0;t)+\bigO(\log^{6} n).$$
Here $t$ is a real parameter, $z_0$ is a regular boundary point of the (connected) Coulomb droplet and $\nu(z_0)$ is the outwards unit normal; the coefficient $C(z_0;t)$ has an apriori structure depending on a number of parameters.

In this note we identify the parameters and obtain a formula for $C(z_0;t)$ in potential theoretic and geometric terms. Our formula holds for a large class of potentials such that the droplet is connected with smooth boundary. Our derivation uses the well known expectation of fluctuations formula. 
\end{abstract}

\maketitle

\section{Introduction and main result} \label{intro} We start by giving some general background on potential theory and the Coulomb gas, in order to conveniently formulate and discuss our main result. We shall be brief and refer to the introduction to \cite{ACC}, to the surveys \cite{BF,D,Me,Fo,ST,S,Z}, and to references therein, for proofs and further details. 

\subsection{Planar Coulomb gas ensembles} 
The Coulomb gas \footnote{We consider only the determinantal case, sometimes known as ``$\beta=2$'', other times as ``$\beta=1$''.} with respect to an external potential $Q:\C\to \R\cup\{+\infty\}$ is an $n$-point configuration $\{z_j\}_1^n\subset \C$ picked randomly with respect to the Gibbs measure 
\begin{align}\label{bogi}
d\Prob_n(z_1,\ldots,z_n)=\frac 1 {Z_n}\prod_{1\le i<j\le n}|z_i-z_j|^2\prod_{j=1}^ne^{-nQ(z_j)}\, dA(z_j),
\end{align}
where we normalize the background measure by $dA(z)=\frac 1 \pi d^2\,z$ for $z\in\C$.

The constant $Z_n$ is chosen so that the measure \eqref{bogi} has unit total mass. For technical reasons we impose the growth condition
\begin{align*}\liminf_{|z|\to\infty}\frac {Q(z)}{2\log|z|}>1,\end{align*}
which is slightly stronger than needed in order for $Z_n$ to exist.
We also assume that $Q(z)$ is lower semicontinuous and finite on some open subset of $\C$.

Under these conditions, the particles tend to distribute according to the law of Frostman's equilibrium measure associated with $Q$, i.e., 
the unique Borel probability measure $\sigma$ which minimizes the functional
$$I_Q[\mu]=\iint_{\C^2}\log\frac 1 {|z-w|}\,d\mu(z)\, d\mu(w)+\int_\C Q\, d\mu$$
over all unit Borel measures $\mu$ on $\C$. More precisely, Johansson's convergence theorem ensures that, with large probability, the random measures $\frac 1 n (\delta_{z_1}+\cdots+\delta_{z_n})$ converge weakly to the equilibrium measure $\sigma$ as $n\to\infty$.

We refer to the support
$$S=S[Q]:=\supp\sigma$$
as the droplet with respect to $Q$. 

The droplet is a compact set which is in general hard to determine; in the following, we assume that it is contained in the interior of the set $\Sigma:=\{z\in\C\,;\,Q(z)<+\infty\}$, and that $Q$ is $C^2$-smooth in the just mentioned interior.

The equilibrium measure then has the apriori structure
$$d\sigma=\Delta Q\cdot\1_S\, dA$$
where we define the normalized Laplacian on $\C$ by
$$\Delta=\d\dbar=\frac 1 4 (\d_x^2+\d_y^2),\qquad (z=x+iy);$$
$\d=\d_x+\frac 1 i \d_y$ and $\dbar=\d_x-\frac 1 i\d_y$ being the usual complex derivatives.

We shall assume that the equilibrium density is nonvanishing on $S$, i.e.,
$$\Delta Q>0\qquad \text{on}\qquad S.$$

The droplet is related to the solution of an obstacle problem in the following way. We define the obstacle function $\check{Q}(z)$ to be the pointwise supremum of $s(z)$ where $s$ runs through the class of subharmonic functions
$s:\C\to\R$ which satisfy $s\le Q$ on $\C$ and $s(w)\le 2\log|w|+\bigO(1)$ as $|w|\to\infty$.

Evidently $\check{Q}(z)$ is a subharmonic function such that $\check{Q}\le Q$ and $\check{Q}(z)= 2\log|z|+\bigO(1)$ as $|z|\to\infty$. Moreover, $\check{Q}(z)$ is globally $C^{1,1}$-smooth and it is harmonic in the complement $\C\setminus S$.

The coincidence set associated with the obstacle problem is the compact set
$$S^*=\{z\,;\,Q(z)=\check{Q}(z)\}.$$

We always have $S\subset S^*$; a connected component of $S^*\setminus S$ is called an outpost, and may have interesting effects for the Coulomb gas, see e.g.~the recent work \cite{AC1}. However, we here assume that there are no outposts, i.e., that $S=S^*$.

In the following, we assume that $S$ is simply connected; we shall
study local properties of the system close to a given point $z_0$ on the boundary $\d S$ of
$S$. (By simple modifications one can obtain analogous results for droplets such that the outer boundary $\d_* S$ \footnote{$\d_* S:=\d U$ where $U$ is the unbounded component of $\C\setminus S$.} is a single, regular Jordan curve and a point $z_0\in \d_* S$, but we do not stress this.) The case of disconnected droplets requires a different analysis and is currently an active area of research, see \cite{ACC,AC1,C1} and references therein.

In order that the boundary $\d S$ be manageable, we require that $Q$ be real-analytic and strictly subharmonic in some neighbourhood of $\d S$. Under this condition, Sakai's regularity theorem implies that the Jordan curve $\d S$ is real-analytic with the possible exception of finitely many singular points (conformal cusps or contact points). We shall assume that there are no singular points, and thus that $\d S$ is an everywhere regular, real-analytic Jordan curve. 

(We refer to the introduction to \cite{ACC} for detailed references to all of the above stated results, which are completely standard in the area.)

\smallskip

\emph{All of the above conditions on $Q$ are tacitly assumed in what follows.}

\subsection{Intensity functions} Let $\{z_j\}_1^n$ be a random sample from \eqref{bogi}.

We define the $1$-point density $R_n=R_{n,1}$ of the system by
$$R_n(w)=\lim_{\eps\to 0}\frac {\E_n(\# D(w;\eps))}{\eps^2},\qquad (w\in\C)$$
where $D(w;\eps)$ is the Euclidean disc $\{z\,;\,|z-w|<\eps\}$ and $\# D(w;\eps):=\#\{j\in\{1,\ldots,n\}\,;\,z_j\in D(w;\eps)\}$ counts the number of particles in the disc.

More generally, we define the $k$-point function $R_{n,k}$ (for $k\le n$) as the unique continuous function on $\C^k$ such that for each continuous and bounded function $f$ on $\C^k$ we have (with $\{z_j\}_1^n$ a random sample)
$$\E_n(f(z_1,\ldots,z_k))=\frac {(n-k)!}{n!}\int_{\C^k}f\cdot R_{n,k}\, dA^{\otimes k}.$$

A basic fact states that the process is determinantal, i.e., there exists a correlation kernel $K_n(z,w)$ such that
$$R_{n,k}(w_1,\ldots,w_k)=\det(K_n(w_i,w_j))_{k\times k}.$$
In particular $R_n(w)=K_n(w,w)$.

The kernel $K_n$ can be taken as the reproducing kernel of the $n$-dimensional subspace $\calW_n\subset L^2(\C,dA)$ of weighted polynomials 
$$\calW_n=\{f(z)=p(z)\cdot e^{-nQ(z)/2}\,;\,p\,\text{is a holomorphic polynomial of degree at most }n-1\}.$$
We equip $\calW_n$ with the usual $L^2$-norm $\|f\|^2=\int_\C|f|^2\, dA$. This choice of correlation kernel is made throughout.

\subsection{Main result}
Let $\nu(z)$ be the unit normal on $\d S$ pointing out of $S$. We define the (signed) curvature $\kappa:\d S\to \R$
by
\begin{equation}\label{curvature}\kappa(z)=\d_\tts \arg \nu(z),\qquad (z\in \d S),\end{equation}
where $\d_\tts$ denotes differentiation with respect to arclength along $\d S$, oriented in the positive sense. We write $ds$ for the arclength measure along $\d S$.

Now fix a boundary point $z_0$ and a real number $t$ with $|t|\le M\sqrt{\log n}$ where $M$ is a suitable, large enough constant (depending only on $Q$). 

It is natural to zoom on $z_0$ in the direction
of $\nu(z_0)$ in the following way
\begin{align}\label{blowup}z=z_0+\frac t {\sqrt{2n\Delta Q(z_0)}}\nu(z_0).
\end{align}

Our main result involves two basic functions $L(z)$ and $L^S(z)$, which appear frequently in the Coulomb gas theory, cf.~in particular \cite{AC1,AHM,ZW}.

We define $$L(z)=\log\Delta Q(z)$$ in some neighbourhood of the droplet, and extend it in some way to a smooth bounded function on $\C$. The function $L^S$ is the Poisson modification, which is continuous on $\C$, equals to $L$ in $S$, and is bounded and harmonic in $\C\setminus S$. 

We also need the complementary error function
$$\erfc t=\frac 2 {\sqrt{\pi}}\int_t^{+\infty}e^{-u^2}\, du.$$

\smallskip

We are now ready to formulate our main result.

\begin{thm} \label{mth} Under the above assumptions, rescaling as in \eqref{blowup}, we have as $n\to\infty$
\begin{align*}R_n(z)=n\Delta Q(z_0)\cdot \frac {\erfc t} 2+\sqrt{n\Delta Q(z_0)}\cdot C(z_0;t)+\bigO(\log^{6} n)\end{align*}
where
$$C(z_0;t)=\frac {\d_\ttn L(z_0)}{\sqrt{2}}
\cdot \frac {t\erfc t}2+\frac {e^{-t^2}}{\sqrt{2\pi}}\bigg(\frac {t^2} 6(\kappa(z_0)-\d_\ttn L(z_0))-\frac 5 {12}\d_\ttn L(z_0)+\frac 1 4\d_\ttn L^S(z_0)-\frac 1 3\kappa(z_0)\bigg).$$

Here $\d_\ttn$ is the normal derivative in the direction of $\nu(z_0)$; the derivative $\d_\ttn L^S(z_0)$ is taken in the exterior of the droplet.

Moreover, the $\bigO$-constant is uniform for $|t|\le M\sqrt{\log n}$.
\end{thm}

Recall that the expansion of the 1-point function in the bulk takes the form (see e.g.~\cite{A,AHM,B,BBS,BF,ZW})
$$R_n(w)=n\Delta Q(w)+\frac 1 2 \Delta L(w)+\bigO(n^{-1}),\qquad (w\in\Int S).$$

Observe that the correction $\frac 1 2 \Delta L$ is by order of magnitude $\sqrt{n}$ smaller than the correction term in Theorem \ref{mth}. In fact, this difference balances the smaller area of the edge regime, so that both terms contribute about equally to the distribution of fluctuations of linear statistics; the leading term in Theorem \ref{mth} plays an equally important role
in this dynamic.
Related ideas appear in \cite{ACC}, where they are used to obtain new results on fluctuations in the disconnected case.

\smallskip

We turn to some other consequences.

\begin{cor} \label{cor1} If $Q(z)$ is a Hele-Shaw potential, viz.~if $\Delta Q$ is constant in a neighbourhood of the droplet, then the coefficient $C(z_0;t)$ reduces to
$$C(z_0;t)=\frac {\kappa(z_0)}6 \frac {e^{-t^2}}{\sqrt{2\pi}}(t^2-2).$$
\end{cor}

 As a special case, we recover the result by Lee and Riser for the elliptic Ginibre ensemble \cite[Theorem 1.1]{LR}, which is thus universal for the class of Hele-Shaw potentials (but not beyond it!). 

\smallskip

We also note the following corollary in the radially symmetric case, which appeared recently in \cite[Theorem 1.11]{ACC}.

\begin{cor} \label{cor2} If $Q(z)$ is radially symmetric, i.e., $Q(z)=Q(|z|)$, then
$$C(z_0;t)=\frac {\d_\ttn L(z_0)}{\sqrt{2}}\cdot \frac {t\erfc t}2+\frac {e^{-t^2}}
{\sqrt{2\pi}}\bigg(\frac {t^2} 6(\kappa(z_0)-\d_\ttn L(z_0))-\frac 5 {12}\d_\ttn L(z_0)-\frac 1 3\kappa(z_0)\bigg).$$
\end{cor}

(To see this, just note that $L^S$ is constant in $\C\setminus S$.) 

\smallskip

In the disconnected case, additional
oscillatory terms enter the subleading edge density; cf.~\cite{ACC} for several very detailed results in this connection.

\smallskip

Theorem \ref{mth} is closely connected to the following expectation of fluctuations formula from \cite{AHM} (we will use a version from \cite{ACC,AC1} which better suits our present notation, and which extends to a context of disconnected droplets). This formula involves the linear statistics
$$\fluct_n f=\sum_{j=1}^n f(z_j)-n\int_\C f\, d\sigma.$$
Here the real-valued test-function $f$ is $C^2$-smooth in some neighbourhood of the droplet; $\sigma$ is the equilibrium measure. The behaviour of $f$ outside a neighbourhood of the droplet is less sensitive; for simplicity we will assume that $f$ is globally smooth and bounded.

\begin{thm} \label{flth} (Cf. \cite{AHM,ACC,AC1}.) As $n\to\infty$ we have the convergence
$$\E_n(\fluct_n f)=\rho_{\frac 1 2}(f)+\bigO(n^{-\beta})$$
where $\beta>0$ is a small enough constant and $\rho_{\frac 1 2}$ is the distribution
$$\rho_{\frac 1 2}(f):=\int_S f\cdot \frac 1 2 \Delta L\, dA-\frac 1 {8\pi}\oint_{\d S}f\cdot \d_{\tt{n}}(L-L^S)\, ds
+\frac 1 {8\pi}\oint_{\d S}\d_{\tt{n}}f\, ds.$$
The normal derivative is taken in the exterior of the droplet.
\end{thm}

\begin{rem} If we instead consider random variables $\trace_n f:=\sum_1^n f(z_j)$ and write $\rho_0(f)=\int f\, d\sigma$,
we obtain the result that $\E_n(\trace_n f)=n\rho_0(f)+\rho_{\frac 1 2}(f)+\cdots$.
The notation $\rho_0,\rho_{\frac 1 2},\rho_1,\ldots$ for distributional correction terms to the density is borrowed from the profound paper \cite{ZW} where an equivalent version of the formula in Theorem \ref{flth} is stated and where also higher correction terms are discussed (a somewhat related computation for the elliptic Ginibre ensemble is given in \cite{LR}). We stress that Theorem \ref{flth} presupposes that the droplet is connected. In the disconnected case, additional terms enter $\rho_{\frac 1 2}$, cf. \cite{ACC3,ACC,AC1,By}. 
\end{rem}

Our strategy is based around an analysis of recent work due to Christophe Charlier in \cite{C}.  
To be more precise, our point of departure is \cite[Theorem 2.4]{C}, which ensures existence of a subleading kernel in the (microscopic) off-diagonal case, where 
\begin{equation}\label{offd}z=z_0+\frac \zeta {\sqrt{2n\Delta Q(z_0)}}\nu(z_0),\qquad w=z_0+\frac \eta {\sqrt{2n\Delta Q(z_0)}}\nu(z_0).\end{equation}
Here $\zeta$ and $\eta$ are arbitrary complex numbers (rescaled variables) with $|\zeta|,|\eta|\le M\sqrt{\log n}$. 

Recall that the correlation kernel $K_n(z,w)$ may be multiplied by a ``cocycle'' without altering the determinantal process. In the present setting, a cocycle is a function of the form $c_n(\zeta,\eta)=u_n(\zeta)\overline{u_n(\eta)}$, where $u_n$ is an arbitrary continuous unimodular function.

As is well-known (cf.\,\cite{HW}) the leading order behaviour of $K_n(z,w)$ under the scaling \eqref{offd} is, up to cocycles given by $n\Delta Q(z_0)$ times the free boundary kernel
$$k(\zeta,\eta):=e^{\frac 1 4 (2\zeta\bar{\eta}-|\zeta|^2-|\eta|^2)}\cdot\frac 1 2 \erfc\bigg(\frac {\zeta+\bar{\eta}}2\bigg).$$

In \cite[Theorem 2.4]{C}, Charlier proves a more detailed expansion, obtaining a subleading kernel $k_2(\zeta,\eta)$ and cocycles $c_n(\zeta,\eta)$ such that
\begin{align}\label{Cthm}c_n(\zeta,\eta)K_n(z,w)=n\Delta Q(z_0)k(\zeta,\eta)+\sqrt{n}k_2(\zeta,\eta)+\bigO(\log^3 n).\end{align}
The expansion \eqref{Cthm} is used in \cite{C} to study the distribution of smallest gaps in random normal matrix ensembles. 

In addition, \cite[Theorem 2.4]{C} gives an apriori structure for the kernel $k_2(\zeta,\eta)$ as a combination of free boundary kernels and suitable ``Gaussian kernels'', where the coefficients are polynomials in the real and imaginary parts of $\zeta$ and $\eta$. The formula contains 18 undetermined real parameters; we shall further study and simplify it in Section \ref{esub}.

\smallskip

Our main objective is however to dissect the diagonal case $k_2(t,t)$ of Charlier's construction and identify the corresponding coefficients. To this end, a key turns out to be the expectation of fluctuations formula in Theorem \ref{flth}. Once the key elements are identified, our computations naturally generalize the approach to connected droplets in the radially symmetric case, found in \cite{ACC}, which serves as another main inspiration.

\begin{rem} 
Recall that the main result on edge asymptotics for correlation kernels from \cite{HW} gives the convergence $c_n(\zeta,\eta)K_n(z,w)=n\Delta Q(z_0)k(\zeta,\eta)+o(n)$ as $n\to\infty$ where $k(\zeta,\eta)$ is the free boundary kernel; cf.~also \cite{AKM}. Recently in \cite[Lemma 1.3]{MMOC} the error term was improved to  $c_n(\zeta,\eta)K_n(z,w)=n\Delta Q(z_0)k_n(\zeta,\eta)+\bigO((\log^3 n)\sqrt{n})$. 
\end{rem}

\begin{rem} 
We do not completely determine the subleading kernel $k_2(\zeta,\eta)$ because of elusive terms which vanish identically along the diagonal $\zeta=\eta$, which have to be fixed by other methods. This has been done in the elliptic Ginibre case, cf.~\cite{BE,Mo}. (On the other hand, a leading order kernel, such as $k(\zeta,\eta)$ is always determined from its diagonal values by a well known polarization procedure, cf. \cite{AKM}.) 
Long range correlations $K_n(z,w)$ along the edge are discussed in the papers \cite{ACC,AC,F2,MMOC} and references therein. The paper \cite{ACC2} gives a comprehensive analysis of correlations on all scales for a class of rotationally symmetric models with hard edge conditions; cf.~also \cite{C1}. In \cite{ACC3,ACC2,ACC,AC1,C1} the droplet is disconnected, giving rise to new oscillatory terms near the boundary which can be expressed in terms of the Jacobi theta function or by so-called $q$-Pochhammer symbols. 
\end{rem}

\section{Existence and apriori structure of the subleading kernel} \label{esub}

In this section we prepare for the proof of Theorem \ref{mth}. We examine the strategy from \cite{C} and incorporate some new additions which will come in handy.

To set things up we fix a boundary point $z_0\in\d S$ and write
\begin{equation}\label{offd1}z=z_0+\frac \zeta {\sqrt{2n\Delta Q(z_0)}}\nu,\qquad w=z_0+\frac \eta {\sqrt{2n\Delta Q(z_0)}}\nu\end{equation}
where $\zeta,\eta$ are complex numbers with $|\zeta|,|\eta|\le M\sqrt{\log n}$ for some large enough constant $M$; $\nu$ is the outwards unit normal at $z_0$.

Let $(e_j)_{j=0}^{n-1}$ be the orthonormal basis for the subspace $\calW_n\subset L^2$ such that
$$e_j=p_j\cdot e^{-nQ/2}$$ 
where the polynomial $p_j(z)$ has degree $j$ and positive leading coefficient.
The reproducing kernel for $\calW_n$ is then
$$K_n(z,w)=\sum_{j=0}^{n-1}e_j(z)\overline{e_j(w)}.$$

By standard estimates, such as in \cite[Section 2]{ACC}, we may for any given $N>0$ ensure that (as $n\to\infty$)
$$K_n(z,w)=\sum_{j=n-C\sqrt{n\log n}}^{n-1}e_{j,n}(z)\overline{e_{j,n}(w)}+\bigO(n^{-N})$$
by choosing $C=C(M)$ large enough.

In order to study asymptotics of the latter sum, we recall a few facts pertaining to $\tau$-droplets 
$$S_\tau:=S[Q/\tau].$$
These droplets form an increasing chain: $\tau<\tau'\Rightarrow S_{\tau}\subset S_{\tau'}$. Moreover, if $\tau$ is close enough to $1$, then the boundary $\d S_\tau$ is an everywhere regular, real analytic Jordan curve. (The evolution of the boundaries $\d S_\tau$ is well studied and is given by weighted Laplacian growth, with respect to the weight $1/2\Delta Q$.) 

For $\tau$ close to $1$, we now form the obstacle function $\check{Q}_\tau(z)$ which increases as $2\tau\log|z|$ near infinity. In detail: 
$\check{Q}_\tau(z)$ is the pointwise supremum of subharmonic functions $s(z)$ which satisfy $s\le Q$ everywhere and $s(w)\le 2\tau\log|w|+\bigO(1)$ as $|w|\to\infty$.

It is a standard fact about Laplacian growth that for $\tau$ close enough to $1$, the $\tau$-droplet equals to the coincidence set, i.e.,
$$S_\tau=\{z\,;\, \check{Q}_\tau(z)=Q(z)\}.$$
(We here use our assumption that this holds for $\tau=1$.)

\smallskip

For $j$ with $n-C\sqrt{n\log n}\le j\le n-1$ we want to approximate the weighted polynomial $e_{j,n}(z)$ at the point
\begin{equation}\label{thepoint}z=z_0+\frac \zeta{\sqrt{2n\Delta Q(z_0)}}.\end{equation}

By the main result in \cite{HW} we have the approximation
\begin{equation}\label{appf}e_{j,n}(z)=(\frac n {2\pi})^{\frac 14}e^{\frac n 2(\calQ_{\tau_j}-Q)(z)+\frac 1 2 \calH_{\tau_j}(z)}\sqrt{\phi'_{\tau_j}(z)}(\phi_{\tau_j}(z))^j\cdot (1+\bigO(n^{-1})),\end{equation}
where the error term is uniform for $|\zeta|\le M\sqrt{\log n}$, while:

\begin{itemize}

\item $\tau_j=j/n$. 

\item $\phi_\tau$ is the conformal map of $\C\setminus S_\tau$ onto the exterior disc $\{z\,;\,|z|>1\}$ which satisfies $\phi_\tau(\infty)=\infty$ and $\phi_\tau'(\infty)>0$; the branch of the square root is chosen so that $\sqrt{\phi_\tau'(\infty})>0$.

\item The function $\calQ_\tau$ is analytic and bounded in the exterior of the curve $\d S_\tau$ and solves the Dirichlet problem $\re\calQ_\tau=Q$ on $\d S_\tau$. We fix the function uniquely by the condition $\im\calQ_\tau(\infty)=0$.

\item The function $\calH_\tau$ is likewise bounded and analytic in the exterior of $\d S_\tau$ but solves instead the Dirichlet problem $\re \calH_\tau=\log\sqrt{\Delta Q}$ on $\d S_\tau$.

\end{itemize}

Note that $\phi_\tau$, $\calQ_\tau$ and $\calH_\tau$ continue analytically across $\d S_\tau$ to a neighbourhood of that curve.

It is convenient to write
$$\calB_\tau:=(2\pi)^{-\frac 1 4}e^{\frac 1 2 \calH_\tau}.$$

We then have the approximation for $1\le j\le C\sqrt{n\log n}$,
\begin{equation}e_{n-j,n}(z)= n^{1/4}e^{\frac n 2(\calQ_{\tau_{n-j}}-Q)(z)}(\phi_{\tau_{n-j}}(z))^{n-j}\sqrt{\phi_{\tau_{n-j}}'(z)}\calB_{\tau_{n-j}}(z)\cdot (1+\bigO(1/n))\end{equation}
uniformly for $|\zeta|\le M\sqrt{\log n}$.

Following \cite{C} we let
$$c_0:=\frac 1 {\sqrt{2\Delta Q(z_0)}}.$$

By Taylor's formula,
$$\sqrt{\phi_{\tau_{n-j}}'(z)}\,\calB_{\tau_{n-j}}(z)=b_0\cdot\bigg(1+b_1\frac {\zeta}{\sqrt{n}}+b_2\frac j n+\bigO(\frac {\log n} n)\bigg),$$
where
\begin{enumerate}[label=(\roman*)]
\item $b_0=\phi'_1(z_0)^{1/2}\calB_1(z_0)$, (so $|b_0|^2=\frac 1 {\sqrt{2\pi}}|\phi_1'(z_0)|\sqrt{\Delta Q(z_0)}$),
\item $b_1=c_0\d_z \log ((\phi_1')^{1/2}\calB_1)(z_0)$, (so $b_1=\frac {c_0}2(\frac {\phi_1''}{\phi_1'}+\calH_1')(z_0)),$
\item $b_2=-\d_\tau\log((\phi_\tau')^{1/2}\calB_\tau(z_0))|_{\tau=1}.$
\end{enumerate}

\smallskip

Let us denote by $V_\tau$ the harmonic continuation of $\check{Q}_\tau|_{\C\setminus S_\tau}$ inwards across $\d S_\tau$.
We then have the basic identity 
\begin{align}\label{vopp}\phi_\tau(z)^{n\tau}e^{\frac n 2(\calQ_\tau-Q)(z)}=
e^{-\frac n 2(Q-V_\tau)(z)}e^{\frac {in} 2(\im\calQ_\tau+2\tau\arg\phi_\tau)(z)}.\end{align}

Following \cite{C}, we analyze the two factors in the right hand side of \eqref{vopp} for $\tau=\tau_{n-j}$, where $j\le C\sqrt{n\log n}$. 

\smallskip

The second factor in \eqref{vopp} is a cocycle for fixed $\tau$, but as $\tau$ varies it must be taken in consideration, especially if we want to study subleading asymptotics.

\smallskip

By Taylor's formula there are real coefficients $a_1,\ldots,a_{20}$ such that, up to negligible terms, for $|\zeta|\le M\sqrt{\log n}$, 
\begin{equation}\label{a:s}
\begin{aligned}&e^{\frac {in} 2(\im \calQ_{\tau_{n-j}}(z)+2\tau_{n-j}\arg \phi_{\tau_{n-j}}(z))} \\
&=
e^{a_1in+i(a_2\re\zeta+a_3\im\zeta)\sqrt{n}+i[a_4j+a_5(\re\zeta)^2+a_6(\re\zeta)(\im\zeta)+a_7(\im\zeta)^2]}\\
&\quad\cdot e^{i(a_8\re\zeta+a_9\im\zeta)\frac j{\sqrt{n}}+ia_{10}\frac {j^2} n}e^{\frac i {\sqrt{n}}(a_{11}(\re\zeta)^3+a_{12}(\re\zeta)^2\im\zeta+a_{13}(\im\zeta)^2\re\zeta+a_{14}(\im\zeta)^3)}\\
&\quad \cdot e^{i(a_{15}(\re\zeta)^2+a_{16}(\re\zeta)(\im\zeta)+a_{17}(\im\zeta)^2)\frac j n+i(a_{18}\re\zeta+a_{19}\im\zeta)\frac {j^2}{n^{3/2}}+ia_{20}\frac {j^3}{n^2}}\cdot \bigg(1+\bigO(\frac {\log^2 n}n)\bigg).
\end{aligned}
\end{equation}

 As observed in \cite{C}, coefficients $a_k$ corresponding to terms which are independent of $j$ represent cocycles that can be canceled without affecting subleading asymptotics, and several other $a_k$ cancel out immediately when considering the product $e_{n-j}(z)\overline{e_{n-j}(w)}$. Of the remaining coefficients only $a_8,a_9$ will ultimately play an essential role for the applications we have in mind. 

 \smallskip

 The first factor in \eqref{vopp} requires a careful analysis. As preparation, we introduce some further notions.

Let $\nu_\tau(z)$ be the outwards unit normal on $\d S_\tau$; we denote by
$\d_{\tts}$ differentiation with respect to arclength on $\d S_\tau$ and $\d_{\ttn}$ differentiation in the direction of $\nu_\tau$. We will abbreviate $\nu=\nu_\tau$ when it is clear from the context that the normal to $\d S_\tau$ is intended.

The operators $\d_{\tts}$ and $\d_{\ttn}$ are vector fields along $\d S_\tau$ which do not commute; computations are frequently simplified by the formulas
\begin{equation}\label{buckly}\d_\ttn=\nu\d+\bar{\nu}\dbar,\qquad \d_\tts=i\nu\d-i\bar{\nu}\dbar.\end{equation}
These formulas are immediate when $\nu=1$ and they follow in general by the change of variables $z\mapsto \nu z$, cf.~also \cite[Appendix A]{ZW}.

We define the argument
$\theta(z):=\arg \nu(z)$ locally,
so that it is smooth (modulo $2\pi$) on the union of the curves $\d S_\tau$ for $\tau$ in some neighbourhood of $1$. Thus the functions $\d_\ttn \theta$ and $\kappa:=\d_\tts\theta$ are well defined and smooth in a neighbourhood of $\d S$ and
$\d_\ttn\theta=0$ identically there.

\smallskip

Following \cite{C}, let now $z_\tau=z_0+\lambda_\tau\nu_1$ be the point on $\d S_\tau$ nearest to $z_0$ on the line $z_0+\nu_1\R$. (So $z_1=z_0$.)

\smallskip

By \cite{HW} we have asymptotic expansions of the form
\begin{align}\label{zcorr}z_{\tau_{n-j}}=z_0+z_0^{(1)}\frac j n+z_0^{(2)}(\frac jn)^2+\cdots\end{align}
and
\begin{align}\label{nucorr}\nu_{\tau_{n-j}}=\nu_1\cdot(1+\nu^{(1)}\frac j n+\nu^{(2)}(\frac jn)^2+\cdots)\end{align}
with suitable coefficients $z_0^{(1)},z_0^{(2)},\ldots$ and $\nu^{(1)},\nu^{(2)},\ldots$, depending on $z_0\in\d S$.

The term $z_0^{(1)}$ is given by
(eg \cite[Lemma 3.3]{AC})
$$z_0^{(1)}=-\frac {|\phi_1'(z_0)|}{2\Delta Q(z_0)}\nu_1.$$
Let us write, for $z\in\d S$ close to $z_0$,
$$\mu(z)=-\frac {z^{(1)}}{\nu_1(z)}=\frac {|\phi_1'(z)|}{2\Delta Q(z)};$$
then the coefficient $\nu^{(1)}$ in \eqref{nucorr} is
\begin{equation}\label{nc2}\nu^{(1)}=i\d_\tts\mu.\end{equation}

It follows that for $1\le j\le C\sqrt{n\log n}$, with $\mu=\mu(z_0)$, $c_0=1/\sqrt{2\Delta Q(z_0)}$,
\begin{equation}\label{split}
\begin{aligned}z:&=z_0+c_0\frac\zeta {\sqrt{n}}\nu_1 \\
&=z_{\tau_{n-j}}+\nu_{\tau_{n-j}}(c_0\frac \zeta {\sqrt{n}}\frac {\nu_1} {\nu_{\tau_{n-j}}}-\frac {z_0^{(1)}}{\nu_{\tau_{n-j}}}\frac j n-\frac {z_0^{(2)}}{\nu_{\tau_{n-j}}}(\frac j n)^2)+\cdots\\
&=z_{\tau_{n-j}}+\nu_{\tau_{n-j}}\cdot \bigg( c_0\frac \zeta {\sqrt{n}}+\mu\frac j n\bigg)+\bigO(\frac {\log^{2} n} n).
\end{aligned}
\end{equation}

\smallskip

We now consider the second factor in \eqref{vopp}. For this purpose we introduce the shorthand notation
\begin{align}F_\tau:=Q-V_\tau.\end{align}
Since the obstacle function $\check{Q}_\tau$ is $C^{1,1}$-smooth, we have $F_\tau=dF_\tau=0$ identically along $\d S_\tau$.

We have the following lemma.

\begin{lem} \label{scot} In a neighbourhood of $\d S$, the Laplacian $\Delta=\d\dbar$ satisfies
$$4\Delta=\d_\tts^2+\d_\ttn^2+\kappa\,\d_\ttn.$$
Moreover, for all $\tau$ in a neighbourhood of $1$, the derivatives $\d_\tts^kF_\tau$ and $\d_\tts^k\d_\ttn F_\tau$ vanish identically along $\d S_\tau$ for all $k\ge 0$ while 
$$\d_\ttn^2 F_\tau=4\Delta F_\tau=4\Delta Q,\qquad (\text{along }\d S_\tau).$$
Finally, with $\tau_{n-j}=\frac {n-j} n$ we have the Taylor expansion, as $\lambda\to 0$,
\begin{align*}F_{\tau_{n-j}}(z_{\tau_{n-j}}+\lambda\nu_{\tau_{n-j}})&=\bigg(2\Delta Q(z_0)-2\mu\d_\ttn\Delta Q(z_0)\frac j n\bigg)(\re\lambda)^2\\
&\quad+\frac 2 3(\d_\ttn-\kappa(z_0))\Delta Q(z_0)(\re\lambda)^3+2\d_\tts\Delta Q(z_0)(\re\lambda)^2\im\lambda+\bigO(|\lambda|^4).
\end{align*}
\end{lem}

\begin{proof} Using the formulas \eqref{buckly} it is straightforward to verify that
$$\d_\ttn^2+\d_\tts^2=4\Delta+2\nu\d\bar{\nu}\dbar+2\bar{\nu}\dbar\nu\d$$
and using that $\d_\ttn\theta=0$ it is easy to verify that
$$\kappa(z)=-2\frac {\bar{\nu}}\nu\dbar\nu=-2\frac {\nu}{\bar{\nu}}\d\bar{\nu}.$$

By \eqref{buckly} we deduce that
$$2\nu\d\bar{\nu}\dbar+2\bar{\nu}\dbar\nu\d=-\kappa\d_n.$$
We have shown that
$4\Delta=\d_\ttn^2+\d_\tts^2+\kappa\d_\ttn$, which can be seen as a generalization of the formula for the Laplacian in polar coordinates; cf.\, also \cite[Appendix A]{ZW}.

Now fix $\tau$ close enough to $1$.

As $V_\tau$ is harmonic and $dF_\tau=0$ along $\d S_\tau$, it follows that
$4\Delta Q=4\Delta F_\tau=(\d_\ttn^2+\d_\tts^2)F_\tau$ along $\d S_\tau$.

Next consider the function
$$\tilde{F}_\tau(s,t)=F_\tau(z(s)+t\nu_\tau(z(s)))$$
where $z(s)$ is an arclength parameterization of $\d S_\tau$ and $t$ is a real parameter. It satisfies $\frac {\d \tilde{F}_\tau}{\d s}|_{t=0}=\d_\tts F_\tau (z(s))$ and $\frac {\d \tilde{F}_\tau}{\d t}|_{t=0}=\d_\ttn F_\tau(z(s))$.

Since
$F_\tau=d F_\tau=0$ along $\d S_\tau$, it follows that there is a smooth function $g_\tau(s,t)$ such that
$$\tilde{F}_\tau(s,t)=t^2g_\tau(s,t).$$
Hence all derivatives of the forms $\d_\tts^k F_\tau$ and
$\d^k_\tts \d_\ttn F_\tau$ for $k\ge 0$ vanish identically along $\d S_\tau$. 

Moreover, 
$$\d_\ttn^3 F_\tau=\d_\ttn(4\Delta-\d_\tts^2-\kappa\d_\ttn)F_\tau=4\d_\ttn\Delta Q-\kappa\cdot 4\Delta Q$$
along $\d S_\tau$.

It is also straightforward to verify that the derivative $\d_\tts\d_\ttn^2 F_\tau$ is independent of the order of differentiation along $\d S_\tau$ (regardless of order, they equal to $2\frac{\d g_\tau}{\d s}(s,0))$ and
$$\d_\tts\d_\ttn^2 F_\tau=\d_\tts(4\Delta-\d_\tts^2-\kappa\d_\ttn)F_\tau=4\d_\tts\Delta Q$$
along $\d S_\tau$. 

By these observations, \eqref{split} and Taylor's formula (applied to $\tilde{F}_{\tau_{n-j}}(s,t)$) we have, as $\lambda\to 0$
\begin{align*}
&F_{\tau_{n-j}}(z_{\tau_{n-j}}+\lambda\nu_{\tau_{n-j}})=2\Delta Q(z_{\tau_{n-j}})(\re\lambda)^2\\
&\qquad\qquad+\frac 1 6(\d_\ttn^3 F_{\tau_{n-j}}(z_{\tau_{n-j}})(\re\lambda)^3
+3\d_\tts\d_\ttn^2 F_{\tau_{n-j}}(z_{\tau_{n-j}})(\re\lambda)^2\im\lambda)
+\bigO(|\lambda|^4) \\
&=(2\Delta Q(z_0)+2\d_\ttn \Delta Q(z_0)\cdot\frac {z^{(1)}}{\nu_1}\frac j n)(\re\lambda)^2\\
&\qquad\qquad+\frac 1 6(4(\d_\ttn-\kappa) \Delta Q(z_0)(\re\lambda)^3+12\d_\tts \Delta Q(z_0)(\re\lambda)^2\im\lambda)
+\bigO(|\lambda|^4).
\end{align*}
The proof is complete. 
\end{proof}

Recalling \eqref{split} and inserting the Taylor expansion in Lemma \ref{scot} with
\begin{equation}\label{sato}\lambda=c_0\frac \zeta{\sqrt{n}}+\mu\frac j n+\bigO(\frac {\log^{2} n} n),\end{equation}
we obtain the following basic result.

To alleviate the notation, we will write throughout:
\begin{align}x:=\frac j {\sqrt{n}}\end{align}
and
\begin{align}\label{consts}\mu:=\frac {|\phi_1'|}{2\Delta Q},\quad \omega:=\frac {|\phi_1'|}{\sqrt{\Delta Q}},\quad c_0:=\frac 1 {\sqrt{2\Delta Q}},\quad a:=\frac {\zeta+\bar{\eta}}{\sqrt{2}},\quad 
\gamma_2:=\frac {z_0^{(2)}}{\nu_1}.\end{align}

\begin{lem} \label{coeffs} With $z$ as in \eqref{thepoint},
we have the Taylor expansion
\begin{align}
\label{right}e^{-\frac n 2 F_{\tau_{n-j}}(z)}&=\exp(-(\frac 1 {\sqrt{2}} \re\zeta+\frac j {\sqrt{n}}\frac \omega 2)^2)\cdot \bigg(1-\frac 1 {2\sqrt{n}}P(\re\zeta,\im\zeta;x)+\bigO(\frac {\log^6 n} n)\bigg),\end{align}
where $$P(u,v;x)=\alpha_{0;\ttn}u^3+\alpha_{0;\tts}u^2 v+[\alpha_{1;\ttn}u^2+\alpha_{1;\tts}uv]x+[\alpha_{2;\ttn}u+\alpha_{2;\tts}v]x^2+\alpha_{3;\ttn}x^3;$$
the coefficients are, with all derivatives evaluated at $z_0$,
\begin{equation}
\begin{aligned}
\alpha_{0;\ttn}&=c_0^3\frac23(\d_\ttn-\kappa)\Delta Q,&\alpha_{0;\tts}&=2c_0^3\d_\tts\Delta Q,\\
\alpha_{1;\ttn}&=-2c_0^2\mu\kappa\Delta Q,&\alpha_{1;\tts}&=4c_0^2\bigl(\mu\,\d_\tts\Delta Q+(\d_\tts \mu)\Delta Q\bigr),\\
\alpha_{2;\ttn}&=-2c_0\mu^2(\d_\ttn+\kappa)\Delta Q(z_0)-4c_0\gamma_2\Delta Q, &\alpha_{2;\tts}&=2c_0\bigl(\mu^2\,\d_\tts\Delta Q+2(\mu\d_\tts \mu) \Delta Q\bigr),\\
\alpha_{3;\ttn}&=\frac23\mu^3(\d_\ttn-\kappa)\Delta Q-2\mu^3\d_\ttn \Delta Q-4\mu\gamma_2\Delta Q.\\
\end{aligned}
\end{equation}
\end{lem}

\begin{proof}
Choose $\lambda$ so that $z$ in \eqref{thepoint} satisfies
$$z=z_{\tau_{n-j}}+\lambda\nu_{\tau_{n-j}};$$
then \eqref{sato} is satisfied and
$$\lambda\nu_{\tau_{n-j}}=z_0-z_{\tau_{n-j}}+c_0n^{-1/2}\zeta\nu_1.$$
Inserting the expansions \eqref{zcorr}, \eqref{nucorr}
we obtain 
$$
\lambda
=n^{-1/2}(c_0\zeta+\mu x)
+n^{-1}\left[-\gamma_2x^2-i(\d_\tts\mu)x(c_0\zeta+\mu x)\right]
+\bigO(n^{-3/2}\log^3 n).$$
Consequently, if $\zeta=u+iv$,
$$
\begin{cases}
\re\lambda
&=n^{-1/2}(c_0u+\mu x)+n^{-1}\left[-\gamma_2x^2+c_0(\d_\tts\mu)xv\right]
+\bigO(n^{-3/2}\log^3 n)\cr
\im\lambda
&=n^{-1/2} c_0v-n^{-1}(\d_\tts\mu)x(c_0u+\mu x)
+\bigO(n^{-3/2}\log^3 n)\cr
\end{cases}.
$$
Also,
$$
\Delta Q(z_{\tau_{n-j}})
=\Delta Q(z_0)-\mu n^{-1/2} x\,\d_\ttn \Delta Q(z_0)+\bigO(n^{-1}\log^2 n).
$$
Substituting these expressions into the Taylor expansion from Lemma \ref{scot} we obtain
$$F_{\tau_{n-j}}(z)=2n^{-1}\Delta Q(z_0)(c_0u+\mu x)^2+n^{-3/2}P(u,v;x)+\bigO(n^{-2}\log^4 n),$$ where $P$ is the
explicit expression in the statement of the lemma.
\end{proof}

We now return to the points $z,w$ in \eqref{offd1} and
consider the products
$$e_{n-j}(z)\overline{e_{n-j}(w)}$$
and sum in $j$, taking in account the above approximations and simplifying.

The end result is that there are 
cocycles $c_n(\zeta,\eta)$ (depending on $z_0$) such that
$$c_n(\zeta,\eta)K_n(z,w)= \sum_{j=1}^{C\sqrt{n\log n}}\sqrt{n}f_0\bigg(\frac j {\sqrt{n}}\bigg)\cdot\bigg[1+\frac 1 {\sqrt{n}}f_1\bigg(\frac j {\sqrt{n}}\bigg)\bigg]+\bigO(\log^3 n),$$
where (with $x=j/\sqrt{n}$)
\begin{align*}
f_0(x)&:=|\phi_1'(z_0)|(2\pi)^{-1/2}\sqrt{\Delta Q(z_0)}\\
&\times\exp\left\{-\frac 1 2 \left(x\omega+\frac 1 {\sqrt{2}}\re (\zeta+\eta)\right)^2-\frac 1 4(\re(\zeta-\eta))^2\right\}\\
&\times  
e^{i(a_8\re(\zeta-\eta)+a_9\im(\zeta-\eta))x}\cdot e^{-\frac i 4\im(\zeta^2+\bar{\eta}^2)}\\
&=|\phi_1'(z_0)|(2\pi)^{-1/2}\sqrt{\Delta Q(z_0)}e^{-\frac 1 4(\re(\zeta-\eta))^2}\cdot e^{-\frac i 4\im(\zeta^2+\bar{\eta}^2)} \\
&\times \exp\left\{-\frac 1 2 \left(x\omega+\frac 1 {\sqrt{2}}\re (\zeta+\eta)-i\frac {a_8}{\omega}\re(\zeta-\eta)-i\frac {a_9}{\omega}\im(\zeta-\eta)\right)^2\right\}\\
&\times \exp\left\{- \frac {i}{\omega\sqrt{2}}(\re(\zeta+\eta))(a_8\re(\zeta-\eta)+a_9\im(\zeta-\eta))
\right\}\\
&\times \exp\left\{-\frac 1 {2\omega^2}\left(a_8\re(\zeta-\eta)+a_9\im(\zeta-\eta)\right)^2\right\},
\end{align*}
and
\begin{align*}f_1(x)&:= (\re b_1)(\re\zeta+\re\eta)-(\im b_1)(\im\zeta+\im\eta)\\
&-\frac 1 2 \alpha_{0;\ttn}
((\re\zeta)^3+(\re\eta)^3)-\frac 1 2 \alpha_{0;\tts}((\re\zeta)^2(\im\zeta)+(\re\eta)^2(\im\eta))
\\
&+x(-\frac 1 2 \alpha_{1;\ttn}((\re\zeta)^2+(\re\eta)^2)-\frac 1 2 \alpha_{1;\tts}((\re\zeta)(\im\zeta)+(\re\eta)(\im\eta))+2\re b_2)\\
&+x^2(-\frac 1 2 \alpha_{2;\ttn}(\re\zeta+\re\eta)-\frac 1 2 \alpha_{2;\tts}(\im\zeta+\im\eta))-\alpha_{3;\ttn}x^3\\
&+i((\re b_1)(\im\zeta-\im\eta)+(\im b_1)(\re\zeta-\re\eta))\\
&+\frac i 2 x(a_{15}((\re\zeta)^2-(\re(\eta)^2))+a_{16}((\re\zeta)(\im\zeta)-(\re\eta)(\im\eta))+a_{17}((\im\zeta)^2-(\im\eta)^2))\\
&+\frac i 2 x^2(a_{18}(\re\zeta-\re\eta)+a_{19}(\im\zeta-\im\eta)).
\end{align*}

Following \cite{C}, we start by simplifying the function $f_0(x)$. Here and throughout, it is convenient to introduce the Ginibre kernel
$$G(\zeta,\eta):=e^{\frac 1 4(2\zeta\bar{\eta}-|\zeta|^2-|\eta|^2)}.$$

\begin{lem} \label{f0} The function $f_0(x)$ simplifies as 
\begin{align*}
f_0(x)&=|\phi_1'(z_0)|(2\pi)^{-1/2}\sqrt{\Delta Q(z_0)}\,G(\zeta,\eta)\, \exp\left\{-\frac 1 2 \left(x\omega+\frac 1 {\sqrt{2}}(\zeta+\bar{\eta})\right)^2\right\}.
\end{align*}

In particular,
\begin{equation}\label{mjugg}-\frac 1 2 f_0(0)=-\frac 1 2 |\phi'(z_0)|\sqrt{\Delta Q(z_0)}\frac {e^{-\frac 1 2 ((\re\zeta)^2+(\re\eta)^2)}}{\sqrt{2\pi}e^{\frac i 4(\im(\zeta^2+\bar{\eta}^2)}}.
\end{equation}
\end{lem}

\begin{proof}

By the Euler-MacLaurin formula we obtain, as $n\to\infty$ 
\begin{align}\label{em}\sum_{j=1}^{C\sqrt{n\log n}}\sqrt{n}f_0(\frac j {\sqrt{n}})=n\int_0^\infty f_0(x)\,dx-\frac 1 2 f_0(0)\sqrt{n}+\bigO(1+|f_0'(0)|).
\end{align}

Using the integral
\begin{align}\int_0^\infty e^{-\frac 1 2 (x\omega+a)^2}\, dx=\frac {\sqrt{2\pi}}\omega \frac 1 2 \erfc(\frac a {\sqrt{2}}),\end{align}
we find
\begin{align*}\sum_{j=1}^{C\sqrt{n\log n}}\sqrt{n}f_0(\frac j {\sqrt{n}})&= n\Delta Q(z_0)\frac 1 2 \erfc\bigg(\frac 1 2 \re(\zeta+\eta)-i\frac {a_8}\omega\re(\zeta-\eta)-i\frac {a_9}\omega \im(\zeta-\eta)\bigg)
\\
&\times \exp\left\{- \frac {i}{\omega\sqrt{2}}(\re(\zeta+\eta))(a_8\re(\zeta-\eta)+a_9\im(\zeta-\eta))
\right\}\\
&\times \exp\left\{-\frac 1 {2\omega^2}\left(a_8\re(\zeta-\eta)+a_9\im(\zeta-\eta)\right)^2\right\}
+\bigO(\sqrt{n}).
\end{align*}

However, we know, by the main results from \cite{HW}, that there are cocycles $c_n$ such that
$$c_n(\zeta,\eta)K_n(z,w)= n\Delta Q(z_0)\frac 1 2 \erfc (\frac 1 2(\zeta+\bar{\eta}))\, G(\zeta,\eta)+o(n)$$

Identifying coefficients, we see that
$$a_8=0,\qquad \text{and}\qquad a_9=-\frac \omega {\sqrt{2}}.$$
It remains to note that
\begin{align}\label{gick}G(\zeta,\eta)e^{-\frac 12 a^2}=e^{-\frac 12 ((\re\zeta)^2+(\re\eta)^2)}e^{-\frac i 4(\im(\zeta^2+\bar{\eta}^2))},\qquad (a:=\frac{\zeta+\bar{\eta}}{\sqrt{2}}),\end{align}
which implies \eqref{mjugg}. \end{proof}

Under the change of variables $y=\omega x$ we obtain
\begin{align}
f_0(x)\, dx=\Delta Q(z_0)\,G(\zeta,\eta)\, \exp\left\{-\frac 1 2\left(y+\frac 1 {\sqrt{2}}(\zeta+\bar{\eta})\right)^2\right\}\, \frac {dy}{\sqrt{2\pi}}.
\end{align}

We next use the Euler-MacLaurin formula to deduce that
\begin{align}\sum_{j=1}^{C\sqrt{n\log n}}f_0(\frac j {\sqrt{n}})f_1(\frac j {\sqrt{n}})&=\sqrt{n}\int_0^\infty f_0(x)f_1(x)\, dx+\bigO((f_0f_1)'(0))\nonumber\\
&=\sqrt{n}\Delta Q(z_0)\,G(\zeta,\eta)\int_0^\infty
f_1(\frac y \omega)e^{-\frac 1 2(y+ \frac {\zeta+\bar{\eta}}{\sqrt{2}})^2}\,\frac {dy}
{\sqrt{2\pi}}+\bigO(\log^2 n),\label{em2}
\end{align}
for $|\zeta|,|\eta|\le M\sqrt{\log n}$.

Write for $j=0,1,2,3$
$$I_j=\int_0^\infty y^je^{-\frac 1 2(y+a)^2}\, \frac {dy}{\sqrt{2\pi}}.$$
We note the following lemma; the proof is straightforward.

\begin{lem}\label{integrals}
We have
\begin{align*}I_0&=\frac 1 2 \erfc(\frac a {\sqrt{2}}),\\
I_1&=\frac {e^{-a^2/2}}{\sqrt{2\pi}}-a \frac 1 2 \erfc(\frac a {\sqrt{2}}),\\
I_2&=-a\frac {e^{-a^2/2}}{\sqrt{2\pi}}+(a^2+1)\frac 1 2\erfc(\frac a {\sqrt{2}}),\\
I_3&=(a^2+2)\frac {e^{-a^2/2}}{\sqrt{2\pi}}-(a^3+3a)\frac 1 2\erfc(\frac a {\sqrt{2}}).
\end{align*}
\end{lem}

Now recall that
\begin{align*}f_1(x)&=C_0+iD_0+(C_1+iD_1)x+(C_2+iD_2)x^2+C_3x^3
\end{align*}
for certain real coefficients $C_j,D_k$ (polynomials in the real and imaginary parts of $\zeta$ and $\eta$); the $D_k$ vanish along the diagonal $\zeta=\eta$.

It follows that, with $\tilde{C}_j=C_j/\omega^j$, $\tilde{D}_j=D_j/\omega^j$
\begin{align*}
\int_0^\infty f_1(\frac y \omega) \, e^{-\frac 1 2(y+a)^2}\, \frac {dy}{\sqrt{2\pi}}&=
(\tilde{C_0}+i\tilde{D}_0)\frac 1 2 \erfc(\frac a {\sqrt{2}})\\
&+(\tilde{C}_1+i\tilde{D}_1)\left(\frac {e^{-a^2/2}}{\sqrt{2\pi}}-a\frac 1 2\erfc(\frac a {\sqrt{2}})\right)\\
&+(\tilde{C}_2+i\tilde{D}_2)\left(-a\frac {e^{-a^2/2}}{\sqrt{2\pi}}+(a^2+1)\frac 1 2 \erfc(\frac a {\sqrt{2}}\right)\\
&+\tilde{C}_3\left((a^2+2)\frac {e^{-a^2/2}}{\sqrt{2\pi}}-(a^3+3a)\frac 1 2\erfc(\frac a {\sqrt{2}})\right).
\end{align*}

We next recall that the free boundary kernel is
$$k(\zeta,\eta):=G(\zeta,\eta)\frac 1 2 \erfc(\frac {{\zeta}+\bar{\eta}}2)$$
and use the identity \eqref{em2}.
Using Lemma \ref{integrals}, we deduce the following result.

\begin{lem} With $z,w$ as in \eqref{offd1},
there are cocycles $c_n(\zeta,\eta)$ depending on $z_0$ such that
$$c_n(\zeta,\eta)K_n(z,w)= n\Delta Q(z_0)k(\zeta,\eta)+\sqrt{n}k_2(\zeta,\eta)+\bigO(\log^3 n),$$
where, with $C_j'=\tilde{C}_j\Delta Q(z_0)=\Delta Q(z_0)C_j/\omega^j$,  and similarly for $D_j'$,
\begin{align*}k_2(\zeta,\eta)&=-\frac 1 2|\phi_1'(z_0)|\sqrt{\Delta Q(z_0)}\frac {e^{-\frac 1 2((\re\zeta)^2+(\re\eta)^2)}}{\sqrt{2\pi}e^{\frac i 4(\im(\zeta^2+\bar{\eta}^2))}}\\
&+(C_0'+iD_0')\,k(\zeta,\eta)\\
&+(C_1'+iD_1')\left(\frac {e^{-\frac 1 2((\re\zeta)^2+(\re\eta)^2)}}{\sqrt{2\pi}\,e^{\frac i 4\im(\zeta^2+\bar{\eta}^2)}}-ak(\zeta,\eta)\right)\\
&+(C_2'+iD_2')\left(-a\frac {e^{-\frac 1 2((\re\zeta)^2+(\re\eta)^2)}}{\sqrt{2\pi}\, e^{\frac i 4\im(\zeta^2+\bar{\eta}^2)}}+(a^2+1)k(\zeta,\eta)\right)\\
&+C_3'\left((a^2+2)\frac {e^{-\frac 1 2((\re\zeta)^2+(\re\eta)^2)}}{\sqrt{2\pi}\, e^{\frac i 4\im(\zeta^2+\bar{\eta}^2)}}-(a^3+3a)k(\zeta,\eta)\right),
\end{align*}
\end{lem}

\begin{rem} Employing a trick at the end of \cite[Appendix A]{C}, one can reduce to the case $D_0'=0$, by multiplying $K_n(z,w)$ by a cocycle of the form $u_n(\zeta)\overline{u_n(\eta)}$ where $u_n(\zeta)=e^{\frac i {\sqrt{n}}(p\re\zeta+q\im\zeta)}$ for suitable real $p$ and $q$.
\end{rem}

Noting that
$$\frac \mu \omega=\frac 1 {2\sqrt{\Delta Q(z_0)}},$$
recalling also the explicit form of the coefficients $C_j$ appearing in the polynomial $f_1(x)$,
we obtain the following explicit formulas.

\begin{lem} \label{chuck} We have, with all derivatives evaluated at $z_0$,
\begin{align*}C_0'&=\Delta Q\cdot((\re b_1)(\re \zeta+\re\eta)-(\im b_1)(\im\zeta+\im\eta))\\
&+\sqrt{2\Delta Q}\bigg(-\frac 1 {12}(\d_\ttn\log \Delta Q-\kappa)\cdot ((\re\zeta)^3+(\re\eta)^3)\\
&\qquad\qquad -\frac 1 4 \d_\tts\log \Delta Q\cdot((\re\zeta)^2\im\zeta+(\re\eta)^2\im\eta)\bigg)\\
C_1'&=\sqrt{\Delta Q}\cdot\bigg(\frac \kappa 4((\re\zeta)^2+(\re\eta)^2)-\frac 1 2(\mu\d_\tts\log\Delta Q+\d_\tts\log\mu)
((\re\zeta)(\im\zeta)+(\re\eta)(\im\eta))\\
&\qquad\qquad\qquad\qquad
+\frac {1}\mu \re b_2\bigg)\\
C_2'&=\sqrt{2\Delta Q}\cdot \bigg(\bigg[\frac 1 8(\d_\ttn+\kappa)\log\Delta Q+\frac 1 4\frac {\gamma_2}{\mu^2}\bigg](\re\zeta+\re\eta)-\frac 1 8\d_\tts\log(2\mu\Delta Q) \cdot(\im\zeta+\im\eta)\bigg),\\
C_3'&=\sqrt{\Delta Q}\cdot \frac 1 {12}\bigg(2\d_\ttn\log\Delta Q+\kappa+6\frac {\gamma_2}{\mu^2}\bigg).
\end{align*}
\end{lem}

Now take $t$ real and set $\zeta=\eta=t$. Then $D_0'=D_1'=D_2'=0$,
$a=\sqrt{2}t$
and $k(t,t)=\frac 1 2\erfc t$. 

Recalling also \eqref{mjugg} and \eqref{em}, we deduce that there are real constants $A,B$ (depending on $z_0$) such that
\begin{equation}\label{chucko}\begin{aligned}k_2(t,t)&=\sqrt{2\Delta Q}\bigg[\bigg(At-\frac {t^3}6(\d_\ttn\log\Delta Q-\kappa)\bigg)\frac {\erfc t} 2-\frac{|\phi'|}{2^{3/2}}\frac {e^{-t^2}}{\sqrt{2\pi}}\bigg]\\
&+\sqrt{\Delta Q}\bigg[\frac {\kappa t^2} 2+B\bigg]\cdot\bigg(\frac {e^{-t^2}}{\sqrt{2\pi}}-\sqrt{2}t\frac {\erfc t} 2\bigg)\\
&+\sqrt{2\Delta Q}\frac t {4}\bigg[(\d_\ttn+\kappa)\log\Delta Q+\frac {2\gamma_2}{\mu^2}\bigg]\cdot\bigg(-\sqrt{2}t\frac {e^{-t^2}}{\sqrt{2\pi}}+(2t^2+1)\frac {\erfc t}2\bigg)\\
&+\sqrt{\Delta Q}\cdot\frac 1 {12}\bigg[2\d_\ttn\log\Delta Q+\kappa+\frac {6\gamma_2}{\mu^2}\bigg]\cdot\bigg((2t^2+2)\frac {e^{-t^2}}{\sqrt{2\pi}}-\sqrt{2}(2t^3+3t)\frac{\erfc t}2\bigg).
\end{aligned}
\end{equation}
(Here $A=\sqrt{2\Delta Q}\re b_1$ and $B=\mu^{-1}\re b_2$.)

The terms containing $t^3\erfc t$ cancel out and
after simplification, we obtain the following result (with new constants $A,B$).

\begin{lem}\label{subs} For $z_0\in\d S$ there are constants $A=A(z_0)$ and $B=B(z_0)$ such that
\begin{align}\label{subl}k_2(t,t)&=\frac {t^2}6\sqrt{\Delta Q(z_0)}\cdot (\kappa-\d_\ttn\log\Delta Q)(z_0)\cdot \frac {e^{-t^2}}{\sqrt{2\pi}}+A\frac {e^{-t^2}}{\sqrt{2\pi}}+Bt\frac{\erfc t}2.\end{align}
Moreover, the coefficients $A(z),B(z)$ depend continuously on the point $z\in\d S$.
\end{lem}

\begin{proof} Only the last statement needs to be proven. But it is clear from \eqref{chucko} that the coefficient $B(z)$ for $\frac{t\erfc t}2$ is continuous, since the functions $\re b_1(z),\im b_1(z),\re b_2(z), \d_\ttn\log\Delta Q(z)$, and so on, are continuous. Likewise $A(z)$ is continuous.
\end{proof}

\section{Proof of the main result} 
We now prove Theorem \ref{mth}.
Consider the functionals
\begin{align*}I_n(f)&:=\E_n(\fluct_n f)-\int_S f\cdot \frac 1 2\Delta L\, dA\\
&\int_\C f(z)\cdot \bigg(K_n(z,z)-\1_S(z)\big(n\Delta Q(z)+\frac 1 2 \Delta\log\Delta Q(z)\big)\bigg)\, dA(z).\end{align*}

By Theorem \ref{flth}, we have for any suitable test function $f$ the convergence
\begin{align}\label{flucthm}\lim_{n\to\infty}I_n(f)=\frac 1 {8\pi}\oint_{\d S}\d_\ttn f\, ds-\frac 1 {8\pi}\oint_{\d S}f\cdot \d_\ttn (L-L^S)\, ds,\end{align}
where $L(z)=\log\Delta Q(z)$ (in a neighbourhood of the droplet) and $L^S(z)$ is the Poisson modification and the normal derivative is taken in the exterior of the droplet.

In view of standard bulk asymptotics (eg \cite{A,BF}) we have that
$I_n(f)=I_n'(f)+o(1)$ as $n\to\infty$ where
$$I_n'(f)=\int_{\mathcal{A}_n}f(z)\bigg(K_n(z,z)-n\Delta Q(z)\1_S(z)\bigg)\, dA(z)$$
and \begin{equation}\label{calan}\mathcal{A}_n=\bigg\{z+\frac t {\sqrt{2n\Delta Q(z)}}\nu(z)\,;\,z\in\d S,\,t\in\R,\, |t|\le M\sqrt{\log n}\bigg\},\end{equation}
provided that the constant $M$ is large enough.

Now set
\begin{equation}\label{nowset}z=z_0+\frac t {\sqrt{2n\Delta Q(z_0)}}\nu(z_0),\qquad (|t|\le M\sqrt{\log n}),\end{equation}
so that (by Lemma \ref{subs})
\begin{align*}K_n(z,z)&=n\Delta Q(z_0)\frac {\erfc t} 2\\
&+\sqrt{n}\bigg[\bigg(\frac {t^2}6\sqrt{\Delta Q(z_0)}(\kappa-\d_\ttn\log\Delta Q)(z_0)+\tilde{A}\bigg)\frac {e^{-t^2}}{\sqrt{2\pi}}+\tilde{B}t\frac{\erfc t}2 \bigg]+\bigO(\log^{6} n),
\end{align*}
for suitable constants $\tilde{A}(z_0)$, $\tilde{B}(z_0)$ which depend continuously on $z_0\in\d S$.

Using that
$$\Delta Q(z)=\Delta Q(z_0)+\frac t {\sqrt{2n\Delta Q(z_0)}}\d_\ttn\Delta Q(z_0)+\bigO(\frac {\log n} n),$$
we see that, for $|t|\le M\sqrt{\log n}$,
\begin{align*}
&K_n(z,z)-n\Delta Q(z)\1_S(z)=n\Delta Q(z_0)\bigg(\frac {\erfc t} 2-\1_{t\le 0}\bigg)\\
&+\sqrt{n}\bigg[-\frac {\d_\ttn\Delta Q(z_0)}
{\sqrt{2\Delta Q(z_0)}}t\1_{t\le 0}+\frac {t^2} 6\sqrt{\Delta Q(z_0)}(\kappa-\d_\ttn\log\Delta Q)(z_0)\frac {e^{-t^2}}{\sqrt{2\pi}}
+\tilde{A}\frac {e^{-t^2}}{\sqrt{2\pi}}+ \tilde{B}t\frac {\erfc t} 2\bigg]\\
&\qquad\qquad+\bigO(\log^3 n).
\end{align*}

Setting here $t=-M\sqrt{\log n}$ and recalling that
$$\frac {\erfc t} 2-1=\frac {e^{-t^2}}{2\sqrt{\pi}t}\cdot(1+\bigO(t^{-2})),\qquad (\text{as}\quad t\to-\infty),$$
cf.~\cite[(7.12.1)]{NIST}, we obtain that 
\begin{equation}\label{sit}\lim_{n\to\infty}\frac {K_n(z,z)-n\Delta Q(z)}{M\sqrt{n\log n}}= \tilde{B}-\frac {\d_\ttn\Delta Q(z_0)}{\sqrt{2\Delta Q(z_0)}},\qquad (z=z_0-\frac {M\sqrt{\log n}}
{\sqrt{2n\Delta Q(z_0)}}\nu(z_0)).\end{equation} 

But in the situation of \eqref{sit} we have $K_n(z,z)-n\Delta Q(z)=\bigO(1)$ as $n\to\infty$ if $M$ is large enough by well known bulk asymptotics in \cite{A}. We infer that
$$\tilde{B}(z_0)=\frac {\d_\ttn\Delta Q(z_0)}{\sqrt{2\Delta Q(z_0)}}.$$

We have arrived at the structure
\begin{align*}
K_n&(z,z)-n\Delta Q(z)\1_S(z)=n\Delta Q(z_0)\bigg(\frac {\erfc t} 2-\1_{t\le 0}\bigg)+\\
&+\sqrt{n\Delta Q(z_0)}\bigg[(\d_n\log\Delta Q(z_0))\frac t {\sqrt{2}}(\frac {\erfc t}2-1)+\frac {t^2} 6(\kappa-\d_\ttn\log\Delta Q)(z_0)\frac {e^{-t^2}}{\sqrt{2\pi}}
+A(z_0)\frac {e^{-t^2}}{\sqrt{2\pi}}\bigg]\\
&\qquad\qquad +\bigO(\log^{6} n),
\end{align*}
with $A(z_0)=\tilde{A}(z_0)/\sqrt{\Delta Q(z_0)}$ a new constant, depending continuously on $z_0$.

We shall now integrate over the domain \eqref{calan}, and for this it is convenient to designate to a point $w\in\calA_n$ the coordinates $(z,t)$ if
\begin{align}\label{cov}w=z+\frac t {\sqrt{2n\Delta Q(z)}}\nu(z),\qquad (z\in\d S,\,|t|\le M\sqrt{\log n}).
\end{align}

In order to integrate in the new coordinates, we prove the following lemma.

\begin{lem} Under the change of variables \eqref{cov}, 
the background measure $dA(w)=\frac 1 \pi d^2 w$ transforms as
$$dA(w)=\frac 1 \pi \frac 1 {\sqrt{2n\Delta Q(z)}}\bigg(1+t\frac {\kappa(z)}{\sqrt{2n\Delta Q(z)}}\bigg)\, dt\, ds(z),$$
where $ds$ is arclength measure along $\d S$ and $\kappa(z)$ is the signed curvature.
\end{lem}

\begin{proof} Let $z(s)$ be a positively oriented arclength parameterization of $\d S$ and write (locally) $\nu(z(s))=e^{i\theta(s)}$, so that $\kappa(z(s))=\theta'(s)$. Also write
$$w(s,t)=z(s)+\frac t {\sqrt{2n\Delta Q(z(s))}}e^{i\theta(s)}.$$
Noting that 
$z'(s)=ie^{i\theta(s)}$, we find that
\begin{align*}\frac {\d w}{\d s}&=
\bigg(1+t\frac {\kappa(z(s))}{\sqrt{2n\Delta Q(z(s))}}\bigg)z'(s)-t\frac {\d_\tts L(z(s))} {\sqrt{8n\Delta Q(z(s))}} e^{i\theta(s)} \end{align*}

Also
$$\frac {\d w}{\d t}=\frac 1 {\sqrt{2n\Delta Q(z(s))}}e^{i\theta(s)}.$$

Since $e^{i\theta(s)},z'(s)$ is a positively oriented orthonormal frame, the real Jacobian of the mapping is
$$J_w(s,t)=\det\begin{pmatrix} -t\frac {\d_\tts L(z(s))} {\sqrt{8n\Delta Q(z(s))}} & 1+t\frac {\kappa(z(s))}{\sqrt{2n\Delta Q(z(s))}}\cr
\frac 1 {\sqrt{2n\Delta Q(z(s))}} & 0
\end{pmatrix},$$
and the claim follows
\end{proof}

Returning to \eqref{nowset}
we now Taylor expand
$$f(z)=f(z_0)+\d_\ttn f(z_0)\frac t {\sqrt{2n\Delta Q(z_0)}}+\bigO(\frac {\log n}n).$$

Using that
$$\int_\R\bigg(\frac {\erfc t}2-\1_{t\le 0}\bigg)\, dt=0$$
and letting $z_0$ vary, we
 see that $$I_n'(f)=\oint_{\d S}T_n(f,z)\, ds(z)+\bigO(\frac {\log^4 n}{\sqrt{n}})$$
where
\begin{align*}T_n(f,z_0)&:=\frac 1 \pi\frac 1 {\sqrt{2n\Delta Q(z_0)}}\sqrt{n\Delta Q(z_0)}\bigg[\frac {\d_\ttn f(z_0)+f(z_0)(\d_\ttn L(z_0)+\kappa(z_0))}{\sqrt{2}}\int_\R t(\frac {\erfc t}2-1)\, dt\\
&\qquad+A(z_0)f(z_0)\int_\R \frac{e^{-t^2}}{\sqrt{2\pi}}\, dt
+\frac 1 6 f(z_0)(\kappa-\d_\ttn\log\Delta Q)(z_0)\int_\R \frac {t^2e^{-t^2}}{\sqrt{2\pi}}\, dt\bigg].
\end{align*}

Evaluating the relevant integrals as in \cite[Section 4]{ACC} we infer that 
\begin{align*}I_n'(f)&=\frac 1 {8\pi}\oint_{\d S}\d_\ttn f \,ds+\frac 1{8\pi}\oint_{\d S}f(z)
\bigg(\d_\ttn L(z)+\kappa(z)+4A(z)+\frac 1 {3}(\kappa(z)-\d_\ttn L(z))\bigg)\, ds(z)\\
&\qquad\qquad +\bigO(\frac {\log^{7} n}{\sqrt{n}}).
\end{align*}

Letting $n\to\infty$ and using the convergence \eqref{flucthm} we obtain, since $A(z)$ is a continuous function on $\d S$ while $f$ is an arbitrary smooth function,
\begin{align*}A(z)&=\frac 1 4(-\d_\ttn L(z)-\kappa(z)-\d_\ttn (L-L^S)(z)-\frac 1 3(\kappa(z)-\d_\ttn L(z)))\\
&=-\frac 5 {12} \d_\ttn L(z)+\frac 1 4 \d_\ttn L^S(z)-\frac 1 {3} \kappa(z),\qquad (z\in\d S).
\end{align*}

For $z$ as in \eqref{nowset} we now obtain
\begin{align*}K_n&(z,z)=n\Delta Q(z_0)\frac {\erfc t} 2+\sqrt{n\Delta Q(z_0)}\bigg[\d_\ttn L(z_0)\cdot \frac t {\sqrt{2}}\frac {\erfc t} 2\\
&\qquad +\frac {e^{-t^2}}{\sqrt{2\pi}}\bigg(\frac {t^2} 6(\kappa(z_0)-\d_\ttn L(z_0))-\frac 5 {12} \d_\ttn L(z_0)+\frac 1 4 \d_\ttn L^S(z_0)-\frac 1 {3} \kappa(z_0)\bigg)\bigg]+\bigO(\log^{6} n),
\end{align*}
finishing our proof of Theorem \ref{mth}. q.e.d.

\end{document}